\documentclass[12pt,reqno,a4paper]{amsart}
\usepackage{amsmath,amssymb,amsfonts}
\usepackage[mathscr]{eucal}

\newtheorem{proposition}{Proposition}
\newtheorem{theorem}{Theorem}
\newtheorem{corollary}{Corollary}

\theoremstyle{definition}
\newtheorem{remark}{Remark}
\newtheorem*{exas}{Examples}
\newtheorem*{exa}{Example}

\def\co{\colon\thinspace}

\newcommand{\cal}{\EuScript}

\renewcommand{\leq}{\leqslant}
\renewcommand{\geq}{\geqslant}

\DeclareMathOperator{\Ber}{Ber}

\DeclareMathOperator{\Mat}{Mat}

\DeclareMathOperator{\tr}{tr} 

\DeclareMathOperator{\str}{str}

\renewcommand{\L}{{\cal L}}

\DeclareMathOperator{\ber}{ber}

\DeclareMathOperator{\Sym}{Sym}

\DeclareMathOperator{\ev}{ev}

\newcommand{\wed}{\wedge}

\newcommand{\RR}{\mathbb R}

\newcommand{\ZZ}{{\mathbb Z}}

\newcommand{\CC}{\mathbb C}
\newcommand{\NN}{{\mathbb N}}

\newcommand{\f}{\mathbf{f}}

\newcommand{\g}{\mathbf{g}}

\def\a{\alpha}
\renewcommand{\b}{\beta}

\newcommand{\ps}{{{\psi}}}

\newcommand{\xp}{{\mathbf{x}}}

\DeclareMathOperator{\alg}{{\cal Alg}}

\title[A short proof of the Buchstaber--Rees theorem]{A short proof of the Buchstaber--Rees theorem}

\author{H.~M.~Khudaverdian}
\author{Th.~Th. Voronov}

\address{School of Mathematics,  University of Manchester, Oxford Road,  Manchester   M13 9PL,  UK}
\email{theodore.voronov@manchester.ac.uk \\ khudian@manchester.ac.uk}


\dedicatory{To the memory of Robin Bullough}

\begin{document}

\maketitle
\begin{abstract} We give a short proof of the Buchstaber--Rees
theorem concerning symmetric powers. The proof is based on the
notion of a formal characteristic function of a linear map of
algebras.
\end{abstract}

\section{Introduction}

This paper is based on a talk given at the Robin K. Bullough Memorial Symposium in June 2009. Robin Bullough had been always interested in new algebraic and geometric  structures arising from integrable systems  theory (and mathematical physics in general) as well as in philosophical understanding of what `integrability' is. We had had numerous conversations on that with him in the last ten years. The subject of this paper is related with both aspects. We use a method inspired by our studies in supermanifold geometry (namely, Berezinians and associated structures, ultimately rooted in quantum physics) to give a short direct proof of a theorem of Buchstaber and Rees concerning symmetric powers of algebras and spaces. Buchstaber and Rees's notion of an ``$n$-homomorphism'', motivated by the earlier studies of $n$-valued groups, cannot be separated from integrable systems in the broad sense. In recent years an understanding of `integrability' of various objects has spread, according to which an `integrable' case of any notion (a system of ODEs, a function, a manifold, ...) is a case somehow distinguished and discrete within the continuum of `generic' (or `non-integrable') cases. It is often related with some non-trivial algebraic identities. An approach to linear maps of algebras that we have put forward (see below), allows to isolate a hierarchy of `good' classes of such maps, which may be also regarded as `integrable'. The  notions of algebra homomorphisms and ``$n$-homomorphisms'' find  their natural place in such a hierarchy based on the analysis of a formal characteristic function of a linear map of algebras that we introduced. We also see the next step of this hierarchy (the ``$p|q$-homomorphisms'', which we hope to study further elsewhere).

In this paper we give a brief exposition of our general method together with its  very concrete application, which is in the title.

A  theory of ``$n$-homomorphisms'' (or ``Frobenius $n$-homomorphisms'')  of algebras was developed in a series of papers of V.~M.~Buchstaber and E.~G.~Rees (see particularly \cite{BuRees2002}, \cite{BuRees2004}, \cite{BuRees2008} and
references therein). We   recall the definition of an $n$-homomorphism in section~\ref{sec.nhom} below.  This notion originated
in the  studies of an analog for multi-valued groups of the Hopf algebra of functions on a group (namely, for an $n$-valued group, the coproduct in such an analog is an $n$-homomorphism) and was then identified
with a structure discovered by Frobenius in his theory of higher
group characters.

The main algebraic result of Buchstaber and Rees is the following
fundamental theorem. For commutative associative algebras with unit
$A$ and $B$ over real or complex numbers (the condition of the commutativity of $A$ can be relaxed),
under some technical assumptions on $B$, there is a one-to-one
correspondence between the algebra homomorphisms $S^nA\to B$ and the
$n$-homomorphisms $A\to B$. Here $S^nA\subset A^{\otimes n}$ is the
symmetric power of $A$ as a vector space with the   algebra
structure induced from $A^{\otimes n}$. In particular, when $A=C(X)$ for a compact Hausdorff
space $X$ and $B=\RR$, this gives the following extension of the
classical  theorem of \cite{GelKol1939}: for any $n$, the symmetric
power $\Sym^nX=X\times \ldots \times X\,/S_n$ is canonically
embedded into the linear space $C(X)^*$ so that the image of the
embedding is the set of all $n$-homomorphisms $C(X)\to \RR$. The
proof of this remarkable result in~\cite{BuRees2002} is a tour de force of combinatorial ingenuity.

In \cite{KhudVo2007a} we suggested the following construction. For
an arbitrary linear map $\f$ of a commutative algebra $A$ into a
commutative algebra $B$ we introduce a formal `characteristic
function' (a formal power series in $z$)
\begin{equation*}
        R_{\f}(a,z)=\exp\left(\f\left(\ln\left(1+az\right)\right)\right)
\end{equation*}
and consider classes of maps $\f$ such that $R_{\f}(a,z)$ is a
genuine function. A formal analysis of the behaviour of $R_{\f}(a,z)$ ``at the infinity'' leads to another crucial notion, of a `Berezinian' on the algebra $A$ associated with a linear map   $\f\co A\to B$ \,---\, shortly, an `$\f$-Berezinian'. The study of the characteristic function and the $\f$-Berezinian allows to single out certain classes of linear maps  `with good properties' among arbitrary linear maps of algebras, as follows.

Suppose the characteristic function $R_{\f}(a,z)$ is a linear function of the variable $z$. This precisely characterizes those linear maps $\f\co A\to B$ that are the ring homomorphisms.

Suppose the characteristic function $R_{\f}(a,z)$ is a polynomial function of the variable $z$. It can be shown that such a condition precisely characterizes the $n$-homomorphisms of algebras  in the sense of Buchstaber and
Rees (for some natural number $n$). Our construction gives  a different approach to their theory. In particular, by using $\f$-Berezinian (which is almost tautologically a multiplicative map), we obtain an effortless proof of the Buchstaber and
Rees main theorem stated above. See  Corollary~\ref{coro.multpsi} from Proposition~\ref{prop.berpsin} below, which is the crucial  point of the proof; the way how it is obtained is the main advantage of our approach. It also leads  to other substantial simplifications of the theory.

The next class of `good' linear maps of algebras $\f\co A\to B$ corresponds to the case when the characteristic function $R_{\f}(a,z)$ is  rational. We call them   the  \emph{$p|q$-homomorphisms}. Geometrically  they correspond to  some interesting generalization  of symmetric powers. This is briefly discussed in the last section. See also in paper~\cite{KhudVo2007b}.

Our approach is inspired by the previous work on invariants of
supermatrices. In paper~\cite{KhudVo2005}  we investigated the rational function
\begin{equation*}
    R_A(z)=\Ber (1+Az)\,,
\end{equation*}
where $A$ is an even operator on a superspace and $\Ber$ is its Berezinian (superdeterminant). It can be also written as $R_A(z)=e^{\str \ln(1+Az)}$ where $\str$ stands for supertrace, due to the relation $\Ber e^X=e^{\str X}\,$.  Comparing the expansions of $R_A(z)$ at zero and at infinity allowed us to establish non-trivial relations for the exterior powers of operators and spaces (in the Grothendieck ring) and in particular obtain a new formula for the Berezinian as the ratio of certain   polynomial invariants.



\section{The formal characteristic function}\label{sec.charfun}
Let $A$ and $B$ be two associative, commutative and unital  algebras
over $\CC$ or $\RR$. We shall study linear maps $\f\co A\to B$ that are \emph{not} assumed to be algebra homomorphisms. For a fixed such map $\f$, we say that an arbitrary function $\phi\co A\to B$ is
\textit{$\f$-polynomial}  if its values $\phi(a)$, where $a\in A$,
are given by a universal (i.e., independent of $a$) polynomial expression in $\f(a)$, $\f(a^2)$, \ldots\,. The ring of $\f$-polynomial functions is naturally graded so that the degree of $\f(a)$ is $1$, the degree of $\f(a^2)$ is $2$, etc.

The \textit{characteristic function} of a linear map of algebras
$\f$  is defined as the  formal
power series with coefficients in $B$
\begin{equation}\label{charfucnt}
    R_{\f}(a,z):=\exp\left(\f\left(\ln\left(1+az\right)\right)\right)=1+\psi_1(a)z+
\psi_2(a)z^2+\psi_3(a)z^3+\dots\,
\end{equation}
(see our paper~\cite{KhudVo2007a}).
It is a ``function'' of both $z$
and $a \in A$. The coefficients $\psi_k(a)$ of the series~\eqref{charfucnt}  are $\f$-polynomial functions
of $a$ of degree $k$, so that $\psi_k(\lambda a)=\lambda^k
\psi_k(a)$. Indeed, by differentiating equation~\eqref{charfucnt} w.r.t. $z$ we
can see that $\psi_k(a)$ can be obtained by  the Newton type
recurrent formulae:
\begin{multline*}
   \psi_1(a) =\f(a)\,, \quad\\
\psi_{k+1}(a)
=\frac{1}{k+1}\left(\f(a)\psi_k(a)-\f(a^2)\psi_{k-1}(a)+\f(a^3)\psi_{k-2}(a)-\ldots+(-1)^{k+1}\f(a^{k+1})\right)\,.
\end{multline*}
With respect to the linear map $\f$, the characteristic function enjoys an obvious   exponential property:
\begin{equation*}
    R_{\f+\g}(a,z)=R_{\f}(a,z)R_{\g}(a,z)\,.
\end{equation*}

For a given linear map $\f$, its characteristic function
obeys the relations
\begin{equation}\label{eq.ident1}
    R_{\f}(a,z)R_{\f}(a',z')=R_{\f}(az+a'z'+aa'zz',1)\,,
\end{equation}
or
\begin{equation}\label{eq.ident2}
    R_{\f}(a,1)R_{\f}(b,1)=R_{\f}(c,1) \quad {\rm if}\quad 1+c=(1+a)(1+b)
\end{equation}
(making sense as formal power series). This directly follows from the  definition. We shall use these relations  later.


One can make the following formal transformation of the
characteristic function of $\f$ aimed at obtaining its `expansion near
infinity'. More precisely:  $R_{\f}(a,z)$ is defined initially  as a formal power series in $z$;
it can be seen as the Taylor expansion at zero of some genuine function
of the real or complex variable $z$  if such a function exists. Assume that it exists and keep the notation
$R_{\f}(a,z)$ for it. We have, by a formal transformation, that near the infinity in $z$,
\begin{multline*}
    R_{\f}(a,z)=e^{\f\ln(1+az)}=e^{\f\ln\left(az(1+a^{-1}z^{-1})\right)}=
    e^{\f\ln(z1)+\f\ln a+\f\ln(1+a^{-1}z^{-1})}=\\
    e^{(\ln z)\cdot \f(1)}e^{\f\ln
    a}e^{\f\ln(1+a^{-1}z^{-1})}=
    z^{\f(1)}\sum_{k\geq 0}e^{\f\ln a}\psi_k(a^{-1})z^{-k}=\\
    e^{\f\ln a}\,z^{\chi}+e^{\f\ln a}\psi_1(a^{-1})\,z^{\chi-1}+e^{\f\ln
    a}\psi_2(a^{-1})\,z^{\chi-2}+\ldots
\end{multline*}
where we have denoted $\chi=\f(1)$. Here we assume
whatever we may need for the calculation, e.g., that $a^{-1}$ exists, and so on. Initially $\f(1)\in B$; an
assumption that there is a Laurent expansion at infinity forces to
conclude that $\chi$ must be a number in $\ZZ$.  We also observe that the formal expression $e^{\f\ln a}$ arises as the coefficient of the leading term at infinity.

\emph{We are not using this heuristic argument in the next sections; however it may be helpful for understanding our approach.} Instead
of discussing how this formal calculation can be made rigorous, we shall go around it  and apply
arguments more specific for a particular case.

\section{From characteristic function to $n$-homomorphisms}\label{sec.nhom}
Suppose  that the formal power series~\eqref{charfucnt}  terminates, i.e.,
$R_{\f}(a,z)$ is a polynomial function in $z$ for all $a\in A$,  and that
the degree of $R_{\f}(a,z)$ is bounded by some $N\in \NN$ independent of $a$.
\textit{We claim  that in this case $\f(1)=n$ where $n$ is a natural number and that
$R_{\f}(a,z)$ is a polynomial of degree $n$, i.e., the degree of $R_{\f}(a,z)$  is at most
$n$ for all $a$ and   is exactly $n$ for some $a$} (provided some technical assumption for the target algebra $B$).

Indeed, let $\chi=\f(1)\in B$. Consider $R_{\f}(1,z)=\exp[\f(\ln(1+z))]=\exp[\chi\ln(1+z)]$. We  show first that the element
$\chi\in B$ is a natural number. We have $\exp[\chi\ln(1+z)]=(1+z)^\chi$
where $(1+z)^\chi$ is considered is a formal power series:
\begin{equation*}
    (1+z)^\chi=1+\chi z+ {\chi(\chi-1)\over 2}+\dots =
   \sum_{k=0}^\infty {\chi (\chi-1)\dots (\chi-k+1)\over k!}z^k\,.
\end{equation*}
But $R_{\f}(a,z)$ is a polynomial of degree at most $N$.
Hence $\chi (\chi-1)\dots (\chi-k+1)=0$ for all $k>N$.
If in an algebra $B$ the equation $b(b-1)(b-2)\ldots(b-k)=0$  implies that
$b=j$ for some $j=0,1,\ldots,k$,   the
algebra $B$ is called {`connected'} (Buchstaber \& Rees 2008). This is satisfied, for example, if $B$ does not have divisors of zero. Provided such a condition for $B$ holds, we conclude that  $\chi=n$ for
some integer $n$ between $1$ and $N$.

Now we show that the value $n=\f(1)\in \NN$ gives the degree of the polynomial function $R_{\f}(a,z)$. For an arbitrary $a$, we apply the above
identity~\eqref{eq.ident1} to obtain
\begin{equation}\label{eq.lhsrhs}
    R_{\f}(a,z)=R_{\f}(za,1)=R_{\f}(z-1,1)R_{\f}\left({1\over z}+a-1,1\right)=z^nR_{\f}\left({1\over z}+a-1,1\right)
\end{equation}
(note that $R_{\f}(z-1,1)=e^{\f(\ln z\cdot 1)}=e^{n\ln z}=z^n$).  More explicitly, the  RHS of~\eqref{eq.lhsrhs} has the form
\begin{equation}\label{eq.expans}
    z^n\left[1+\f\left(\frac{1}{z}+a-1\right)+\psi_2\left({\frac{1}{z}}+a-1\right) +\dots+
  \psi_N\left({1\over z}+a-1\right)\right]\,.
\end{equation}
Since the functions $\psi_k(a)$ are $\f$-polynomial
of degrees $k$, all terms in the bracket are inhomogeneous polynomials in $z^{-1}$ of degrees $0, 1, \ldots, N$ respectively.
We are given that $R_{\f}(a,z)$ is a polynomial in $z$; it follows that $N\leq n$.
By multiplying in~\eqref{eq.expans}  through and comparing with the expansion  in~\eqref{charfucnt}, we conclude  that the degree of $R_{\f}(a,z)$ in $z$ is at most $n$,
\begin{equation*}
    R_{\f}(a,z)=1+\psi_1(a)z+ \psi_2(a)z^2+\psi_3(a)z^3+\dots
+\psi_n(a)z^n\,,
\end{equation*}
for any $a$. In particular, for $a=1$, we have $R_{\f}(1,z)=(1+z)^n$ where the degree is exactly $n$. This completes the proof of the claim above.


Let us consider a linear map $\f\co A\to B$   such that its characteristic function $R_{\f}(a,z)$ is a polynomial of degree $n$ in $z$, i.e., it is at most $n$ for all $a$ and it is exactly $n$ for some $a$. As we have found, the integer $n$ is necessarily the value of $\f$ at $1\in A$.

We shall show now that  the class of such maps   coincides with the class of $n$-homomorphisms introduced by Buchstaber and Rees.
Let us   recall their definition. For a given linear map $\f\co A\to B$, Buchstaber and Rees defined  maps
\begin{equation*}
    \Phi_k\co A\times \ldots \times A \to B\,,
\end{equation*}
for all $k=1,2, \ldots\,$,  by a ``Frobenius recursion  formula'':  $\Phi_1=\f$, and
\begin{multline*}
    \Phi_{k+1}(a_1,\ldots,a_{k+1})=\f(a_1)\Phi_k(a_2,\ldots,a_{k+1})
    -\Phi_k(a_1a_2,\ldots,a_{k+1})-\ldots
    -\Phi_k(a_2,\ldots,a_1a_{k+1})\,.
\end{multline*}

A  linear map $\f\co A\to B$ is called a \textit{(Frobenius) $n$-homomorphism}
if $\f(1)=n$ and $\Phi_k  =0$ for all $k\geq n+1\,$ (\cite{BuRees1997, BuRees2002}).

\begin{proposition} The class of the linear maps $\f\co A\to B$   such that their characteristic functions $R_{\f}(a,z)$ are polynomials of degree $n$ in $z$ coincides with the class of the Buchstaber--Rees Frobenius $n$-homomorphisms.
\end{proposition}
\begin{proof}
From the  recursion formula, it is easy to show by induction
that the multilinear maps $\Phi_k$ are symmetric. Therefore they are
defined by the restrictions to the diagonal. Using
induction again, one deduces that the functions
$\varphi_k(a)=\Phi_k(a,\ldots,a)$ obey Newton-type recurrence
relations similar to those satisfied by our functions $\psi_k(a)$ defined as the coefficients of the expansion~\eqref{charfucnt}. From here one can establish that $\Phi_k(a,\ldots,a)=k!\psi_k(a)$, so $\Phi_k(a_1,\ldots,a_k)$
can be recovered from $\psi_k(a)$ by  polarization (see remark below). Therefore the identical vanishing of the Frobenius maps $\Phi_k$ for all $k\geq n+1\,$ is equivalent to the characteristic function $R_{\f}(a,z)$ being a polynomial of degree $\leq n$. Suppose the characteristic function $R_{\f}(a,z)$ is a polynomial of degree  $n$, i.e., of degree $\leq n$ for all $a$ exactly $n$ for some $a$. Then $\f(1)=\chi$ is an integer between $1$ and $n$; if it is less than $n$, then the degree of $R_{\f}(a,z)$ is less than $n$ for all $a$ as shown above. Hence $\f(1)=n$. Conversely, if $\f(1)=n$ , then $R_{\f}(1,z)=(1+z)^n$ and $R_{\f}(a,z)$ is a polynomial  of degree $n$ as claimed.
\end{proof}

Therefore we can identify the   $n$-homomorphisms as defined in \cite{BuRees1997}, \cite{BuRees2002}   with the linear maps such  that their characteristic functions are polynomials of degree $n$.

\begin{remark} Here is a formula for the polarization of a homogeneous polynomial
of degree $k$ (the restriction of a symmetric $k$-linear function to
the diagonal):
\begin{equation}\label{eq.polar}
    \Phi_k(a_1,a_2,\dots,a_k)=
              \sum_{r=1}^k\sum_{1\leq i_1<i_2<\dots<i_r\leq
k}(-1)^{k+r}
           \psi_k(a_{i_1}+a_{i_2}+\dots+a_{i_r})\,.
\end{equation}
Here $\Phi_k(a,\ldots,a)=k!\psi_k(a)$. For example,
             $$
              \Phi_2(a_1,a_2)=\psi_2(a_1+a_2)-\psi_2(a_1)-\psi_2(a_2)
             $$
and
\begin{multline*}
    \Phi_3(a_1,a_2,a_3)=\psi_3(a_1+a_2+a_3)-\psi_3(a_1+a_2)-\psi_3(a_1+a_3)-\psi_3(a_2+a_3)\\
    +
  \psi_3(a_1)+\psi_3(a_2)+\psi_3(a_3)\,.
\end{multline*}
The relation between homogeneous polynomial functions and symmetric multilinear functions is standard, but explicit formulas are not easy to find in the literature.
\end{remark}

\begin{remark} \label{rem.genfun} It may be noted that the `ghost' of our characteristic function $R_{\f}(a,z)$ (to which we came motivated by the study of Berezinians in~\cite{KhudVo2005}) did appear  in relation with $n$-homomorphisms but was never recognized. We can see these instances in hindsight. The  initial definition of an ``$n$-ring homomorphism'' in  \cite{BuRees1996} (with  a bit different normalization from that adopted later) used  a certain monic polynomial $p(a,t)=t^n-\ldots\,$ of degree $n$.  That definition was quickly abandoned  starting from \cite{BuRees1997} in  favour of the Frobenius recursion. In  hindsight one can relate the Buchstaber and Rees polynomial $p(a,t)$  with our characteristic function when it is a polynomial of a given degree $n$, by the formula  $p(a,t)=t^nR_{\f}(a,-\frac{1}{t})\,$. In \cite[Corollary 2.12]{BuRees2002},  the maps $\Phi_k(a,\ldots,a)$ for an $n$-homomorphism were assembled into a certain generating function, which was then re-written in the exponential form. The resulting power series can be identified with our characteristic function.  That series was  used only in the proof of their Theorem~2.9  about the sum of   $n$- and    $m$-homomorphisms, while the  central Theorem 2.8 concerning the relation of the Frobenius $n$-homomorphisms  with the algebra homomorphisms of the symmetric powers, was obtained in \cite{BuRees2002} by a long   combinatorial argument.
\end{remark}

\section{Berezinian and a proof of the key statement}

Let $\f\co A\to B$ be an arbitrary linear map of algebras. We define formally the \textit{$\f$-Berezinian} on $A$ as a map $\ber_{\f}\co
A\to B$ by the formula
\begin{equation*}
    \ber_{\f}(a):=\exp \f(\ln a)=R_{\f}(a-1,1)
\end{equation*}
when
it makes sense. (Compare the Liouville formulas for   matrices, $\det  X=\exp {\tr\ln X}$ and $\Ber X=\exp{\str \ln X}$, with  the ordinary and super determinants, respectively.) Note that the $\f$-Berezinian $\ber_{\f}(a)=e^{\f \ln a}$ appeared in the heuristic calculation in section~\ref{sec.charfun} above  as the leading coefficient    of the characteristic function  $R_{\f}(a,z)$ ``at infinity''.

\begin{theorem}\label{prop.bermult}
{The function $\ber_{\f}$
is multiplicative:}
\begin{equation*}
    \ber_{\f}(a_1a_2)=\ber_{\f}(a_1)\ber_{\f}(a_2)\,
\end{equation*}
(whenever both sides make sense).
\end{theorem}
\begin{proof} By the definition:
\begin{equation*}
    \exp \f(\ln (a_1a_2))=\exp\f\left(\ln a_1+\ln a_2\right)=\exp\left(\f\ln
a_1+\f\ln a_2\right)=\exp \f(\ln a_1) \exp \f(\ln a_2)\,.
\end{equation*}
\end{proof}
This holds even if the algebra $B$ is non-commutative, but $\f$ is a `trace', i.e., $\f(a_1a_2)=\f(a_2a_1)$.

Let $\f$ be an $n$-homomorphism. Then
$\ber_{\f}(a)=\exp \f(\ln a)$ is an $\f$-polynomial function of $a$ with
values in $B$:
\begin{equation*}
    \ber_{\f}(a)=\exp \f(\ln
(1+(a-1)))=R_{\f}(a-1,1)=\\
1+\f(a-1)+\psi_2(a-1)+\dots+\psi_n(a-1)\,,
\end{equation*}
and is well-defined for all $a$.
\begin{proposition} \label{prop.berpsin}
For an $n$-homomorphism,
\begin{equation*}
    \ber_{\f}(a)=\psi_n(a)\,.
\end{equation*}
\end{proposition}
\begin{proof}
Consider  the equality
\begin{multline*}
    1+\psi_1(a)z+ \psi_2(a)z^2+\psi_3(a)z^3+\dots
+\psi_n(a)z^n=\\
z^n\left[1+\f\left({1\over z}+a-1\right)+\psi_2\left({1\over
z}+a-1\right) +\dots+
  \psi_n\left({1\over z}+a-1\right)\right]
\end{multline*}
(compare formulas~\eqref{eq.lhsrhs}, \eqref{eq.expans};  we have legitimately set $N=n$).
Collecting all the terms of degree $n$ in $z$, we arrive at the identity
\begin{equation*}
    \psi_n(a)=1+\f\left(a-1\right)+\psi_2\left(a-1\right) +\dots+
  \psi_n\left(a-1\right)\,,
\end{equation*}
where the RHS is by the definition $\ber_{\f}(a)$.
\end{proof}

This proposition is the crucial step in our proof of the Buchstaber--Rees theorem.

\begin{corollary} \label{coro.multpsi} For an  $n$-homomorphism $\f$, the function
$\psi_n(a)$ is multiplicative in $a$.
\end{corollary}

The multiplicativity of the  function $\psi_n(a)$ for $n$-homomorphisms is the central fact in the Buchstaber--Rees theorem. Establishing this fact was the main difficulty of the proof in~\cite{BuRees2002},  where it was deduced by complicated combinatorial  arguments. In our approach this fact comes about almost without effort.

\begin{remark}
The apparatus of characteristic functions allows to obtain easily
many other facts. For example, if $\f$ and $\g$ are $n$- and
$m$-homomorphisms $A\to B$, respectively, then the exponential
property of characteristic functions immediately implies that
$\f+\g$ is an $(n+m)$-homomorphism, since its characteristic
function is the product of polynomials of degrees $\leq n$ and $\leq
m$. If $\g$ is an $m$-homomorphism $A\to B$ and $\f$ is an
$n$-homomorphism $B\to C$, then $R_{\f\circ
\g}(a,z)=e^{\f\g\ln(1+az)}=e^{\f\ln R_{\g}(a,z)}=\ber_{\f} R_{\g}(a,z)$.
Since we know that $R_{\g}(a,z)$ is a polynomial in $z$ of degree at
most $m$, and the $\f$-Berezinian $\ber_{\f} b$ is a polynomial in
$b\in B$ of degree $n$, we conclude that $R_{\f\circ \g}(a,z)$ has
degree at most $nm$ in $z$, therefore $\f\circ \g$ is an
$nm$-homomorphism. (The first statement  was  established in~\cite{BuRees2002} by a similar argument,  see Remark~\ref{rem.genfun}, while the second statement was obtained in paper~\cite{BuRees2008} in a much harder way  as a corollary of the  main theorem.)
\end{remark}

\section{A completion of the proof}

Let us formulate the main theorem of Buchstaber and Rees.

\begin{theorem}[\cite{BuRees2002}]
There is a one-to-one
correspondence between the $n$-homomor\-phisms   $A\to B$ and the
algebra homomorphisms $S^nA\to B$. The algebra homomorphism $F\co S^nA\to B$ corresponding to an $n$-homomor\-phism    $\f\co A\to B$ is given by the formula
\begin{equation}\label{eq.br}
    F(a_1, \ldots,  a_n)=\frac{1}{n!}\Phi_n(a_1,\ldots, a_n)\,,
\end{equation}
where in the left-hand side a linear map from $S^nA$ is written as a symmetric multilinear function.
\end{theorem}

Here $\Phi_n(a_1,\ldots, a_n)$ is the top non-vanishing term of the Frobenius recursion for $\f$.

Basing on the key result established in the previous section (Corollary~\ref{coro.multpsi}), we can complete the proof of this theorem as follows.


Let ${\alg}^n(A,B)$ be the set of all $n$-homomorphisms from an
algebra $A$ to an algebra $B$.  We shall construct two mutually
inverse maps between the spaces ${\alg}^n(A,B)$ and ${\alg}^1(S^nA,B)$, thus establishing their one-to-one correspondence:
\begin{equation*}
    {\alg}^1(S^nA,B)\underset{\beta}{\overset{\a}{\rightleftarrows}} {\alg}^n(A,B)\,.
\end{equation*}
To every  algebra homomorphism $F\in {\alg}^1(S^nA,B)$ we shall assign
an $n$-homomorphism $\a(F)\in {\alg}^n(A,B)$, and to every $n$-homomorphism $\f\in {\alg}^n(A,B)$ we shall assign an algebra homomorphism $\b(\f)\in {\alg}^1(S^nA,B)$, in the following way.

It
is convenient to introduce an   $n\times n$ matrix $\L(a)$ with
entries in the algebra $A^{\otimes n}=A\otimes \ldots \otimes A$, where
\begin{equation*}
    \L(a)={\rm diag}\,\left[a\otimes 1\otimes \dots \otimes 1,\,
                   1\otimes a\otimes 1\otimes\dots\otimes 1,\dots,
                        1\otimes 1\otimes\dots\otimes 1\otimes a
                        \right]\,.
\end{equation*}
The map $a\mapsto \L(a)$ is a matrix representation $A\to \Mat(n,A^{\otimes n})$. Consider an equation:
\begin{equation}\label{eq.det}
   \boxed{F\bigl(\det (1+ \L(a)z)\bigr)=R_{\f}(a,z)\,. }
\end{equation}
The coefficients of the determinant in the left-hand side take  values \emph{a priori} in the algebra $A^{\otimes n}$, but they actually belong to the subalgebra $S^nA$.
Let~\eqref{eq.det}   hold identically in $z$. We shall show that, for a given $F\,$, equation~\eqref{eq.det} uniquely defines $\f$ so that we can set $\a(F):=\f\,$, and conversely, for a given $\f$, equation~\eqref{eq.det} uniquely defines $F$ so that we can set $\b(\f):=F\,$. (Then the maps $\a$ and $\b$ will automatically be mutually inverse.)

Suppose $F\in {\alg}^1(S^nA,B)$ is given. Then, by comparing the linear terms in~\eqref{eq.det}, we see that $\f$ should be given by the formula
\begin{equation} \label{eq.trace}
    \f(a)=F(\tr \L(a))\,.
\end{equation}
We have $\tr\L(a)=a\otimes 1\otimes \dots \otimes 1 +  \ldots +    1\otimes 1\otimes\dots\otimes a\in S^nA$.
\begin{remark} The element $a\otimes 1\otimes \dots \otimes 1 +  \ldots +    1\otimes 1\otimes\dots\otimes a$ appears in Buchstaber \& Rees (2008) where it is denoted $\Delta (a)\,$.
\end{remark}
Take~\eqref{eq.trace} as the definition of $\f$.  Evidently,  it is a linear map $A\to B\,$. Calculate its characteristic function. We have
\begin{multline*}
    R_{\f}(a,z)=e^{\f\ln(1+az)}=e^{F(\tr \L(\ln(1+az)))}=e^{F(\tr \ln(1+\L(a)z))}
    = F\!\left(e^{\tr \ln(1+\L(a)z)}\right) \\
    =F\left(\det (1+\L(a)z)\right)\,
\end{multline*}
(note that both $\L$ and $F$ can be swapped with functions, as we did in the calculation above, because they are algebra homomorphisms).  So equation~\eqref{eq.det} is indeed satisfied. In particular, the characteristic function of $\f$ is a polynomial of degree $n$. Hence $\f$ is an $n$-homomorphism $A\to B$. We have constructed the desired map $\a\co {\alg}^1(S^nA,B)\to  {\alg}^n(A,B)\,$.

Suppose now $\f\in {\alg}^n(A,B)$ is given. We wish to define $F$ from~\eqref{eq.det}.  We need to show  the existence of a linear map $F\co S^nA\to B$ and that it is an algebra homomorphism. Assume that a linear map $F$ satisfying~\eqref{eq.det} exists. We shall deduce its uniqueness.  By developing the left-hand side of~\eqref{eq.det}, we see that this equation specifies $F$ on all the elements of $S^nA$ of the form $\tr\wed^k\L(a)\in S^nA$,
including   $\det \L(a)$. In particular, we should have
\begin{equation}\label{eq.det2}
    F(\det\L(a))=\ber_{\f}(a)\,.
\end{equation}
Take~\eqref{eq.det2} as the definition of $F$ on such elements.
Note that the elements of the form~$\det\L(a)=a\otimes\ldots\otimes a\in S^nA$ linearly span $S^nA$, so if a linear map $F\co S^nA\to B$ satisfying~\eqref{eq.det2} exists, this formula determines it uniquely. Moreover, by replacing $a$ by $1+az$ in~\eqref{eq.det2}, we see that equation~\eqref{eq.det} will be automatically satisfied. For the existence, observe that $\ber_{\f}(a)=\psi_n(a)$ and $\psi_n$ is the restriction on the diagonal of the symmetric multilinear map $\frac{1}{n!}\Phi\co A\times \ldots \times A\to B$, which  corresponds to a linear map $S^nA\to B$.  It is the map $F$  we are looking for.    (We recover the Buchstaber--Rees formula~\eqref{eq.br}.) It remains to  check that so obtained $F$ is indeed an algebra homomorphism. This immediately follows from Corollary~\ref{coro.multpsi}, which tells that $F$ is multiplicative on the elements $a\otimes\ldots\otimes a$,  where it is $\psi_n(a)$, because the elements $a\otimes\ldots\otimes a$   span  the whole algebra $S^nA\,$. So  we have arrived at the desired map $\b\co {\alg}^n(A,B)\to {\alg}^1(S^nA,B)\,$, and the maps $\a$ and $\b$ are mutually inverse by the construction.

This concludes the
proof.

\section{Rational characteristic functions and $p|q$-homomorphisms}

Our notion of characteristic function readily suggests the next step in the investigation of linear maps of algebras with good properties (after the $n$-homomorphisms). We shall consider it now briefly.

Suppose the characteristic function $R_{\f}(a,z)$ of a linear map $\f\co A\to B$ is not a polynomial in $z$,
but a rational function that can be written as the
ratio of polynomials of degrees $p$ and $q$.
We call such a linear map a \textit{$p|q$-homomorphism}. One can deduce that for a $p|q$-homomorphism  the value $\f(1)$ must be an integer: $\chi=\f(1)=p-q$.

\begin{exas} The negative $-\f$ of a ring homomorphism
$\f$ is a $0|1$-homomorphism.  The difference $\f_{(p)}-\f_{(q)}$ of
a $p$-homomorphism $\f_{(p)}$ and a $q$-homomorphism $\f_{(q)}$ is a
$p|q$-homomorphism. In particular, a linear combination of algebra
homomorphisms of the form $\sum n_{\a} \f_{\a}$ where $n_{\a}\in
\ZZ$ is a $p|q$-homomorphism with $\chi=\sum n_{\a}$,
$p=\sum\limits_{n_{\a}>0} n_{\a}$, and $q=-\sum\limits_{n_{\a}<0}
n_{\a}$.  (This  follows from the exponential property of the
characteristic function.)
\end{exas}

The geometric  meaning of  $p|q$-homomorphisms is related with a certain generalization of the  notion of symmetric powers.

Consider a  topological space $X$. We define its  \textit{$p|q$-th
symmetric power} $\Sym^{p|q}(X)$  as the identification space of
$X^{p+q}=X^p\times X^q$ with respect to the action of the group $S_p\times S_q$ together with the
relations
$$(x_1,\ldots,x_{p-1},y,x_{p+1}\ldots,x_{p+q-1},y)\sim
(x_1,\ldots,x_{p-1},z,x_{p+1}\ldots,x_{p+q-1},z)\,.
$$
The algebraic analog of $\Sym^{p|q}(X)$ is the \textit{$p|q$-th
symmetric power}  of a commutative associative algebra
with unit $A$, which we denote $S^{p|q}A$. We define the algebra $S^{p|q}A$ as the subalgebra
$\mu^{-1}\left(S^{p-1}A\otimes S^{q-1}A\right)$ in $S^pA\otimes
S^qA$, where $\mu\co S^pA\otimes S^qA \to S^{p-1}A\otimes
S^{q-1}A\otimes A$ is the multiplication of the last arguments.

\begin{exa} For the algebra of polynomials in one variable $A=\CC[x]$, it can be shown that the algebra $S^{p|q}A$ is the algebra of
all polynomial invariants of $p|q$ by $p|q$ matrices.  In more detail, consider even matrices $p|q$ by $p|q$ where the entries are regarded as indeterminates of appropriate parity (over complex numbers). The general linear supergroup $GL(p|q)$ acts on such matrices by conjugation. `Polynomial invariants' of $p|q$ by $p|q$ matrices are the polynomial functions of the matrix entries invariant under such action. Denote their algebra by $I_{p|q}$. The problem is to describe the algebra $I_{p|q}$ in terms of functions of the $p+q$  eigenvalues $\lambda_i,\mu_{\a}$ similar to the classical case $q=0$. The restriction of a polynomial invariant of $p|q$ by $p|q$ matrices to the diagonal matrices is clearly a $S_p\times S_q$-invariant polynomial (or an element of $S^pA\otimes S^qA$ where $A=\CC[x]$). Such a polynomial $f(\lambda_1\,\ldots,\lambda_p,\mu_1,\ldots,\mu_q)$ separately symmetric in $\lambda_i$ and $\mu_{\a}$ extends to an element of $I_{p|q}$ if and only if it satisfies an extra condition which is equivalent to $f\in S^{p|q}A$\,. (Here there is a great difference with the rational invariants, for which no extra condition arises.) This is a
non-trivial theorem that can be traced to \cite{Ber1977} (see \cite[p. 294]{Ber1987}; see also \cite{Ser1982}.  See discussion in~\cite{KhudVo2005}.
\end{exa}

(The algebra described in the  example above and its deformations have recently become important   in integrable systems, see, e.g.,  \cite{SerVes2004}.)

There is a relation between the algebra homo\-morphisms $S^{p|q}A\to B$
and $p|q$-homomorphisms $A\to B$. To each homomorphism $S^{p|q}A\to
B$ canonically corresponds a $p|q$-homomorphism $A\to B$.

\begin{exa}
For the algebra of functions on a topological space $X$, the element $\xp=[x_1,\ldots,x_{p+q}]\in \Sym^{p|q}(X)$ defines a
$p|q$-homomorphism   $\ev_{\xp}\co C(X)\to \RR$ by the formula
\begin{equation*}
    a\mapsto a(x_1)+\ldots +a(x_{p})-\ldots
 -a(x_{p+q})\,.
\end{equation*}
This gives a natural map from $\Sym^{p|q}(X)$ to the dual space $A^*$ of the algebra $A=C(X)$,
which generalizes the Gelfand--Kolmogorov  map $X\to A^*$ and the Buchstaber--Rees  map $\Sym^n(X)\to A^*$.
\end{exa}

By using formulas from~\cite{KhudVo2005}, the condition that $\f\co A\to
B$ is a $p|q$-homomorphism can be expressed by the algebraic  equations
\begin{equation}\label{eq.vanhank}
\f(1)=p-q \quad \text{and} \quad
\begin{vmatrix}
      \ps_k(\f,a)  & \dots & \ps_{k+q}(\f,a)  \\
      \dots & \dots & \dots \\
      \ps_{k+q}(\f,a)  & \dots & \ps_{k+2q}(\f,a)  \\
    \end{vmatrix}=0
\end{equation}
for all $k\geq p-q+1$ and all $a\in A$, where $\psi_k(\f,a)=\ps_k(a)$ are the coefficients in the expansion of the characteristic function~\eqref{charfucnt}. The determinants arising in~\eqref{eq.vanhank} are the well-known Hankel determinants. The condition that they identically vanish   replaces the condition $\ps_k(a)=0$ for all $k\geq n+1$ for an $n$-homomorphism.

We see that the image of $\Sym^{p|q}(X)$ in $A^*$, where $A=C(X)$, under the map $\xp\mapsto \f=\ev_{\xp}$
satisfies equations~\eqref{eq.vanhank}. The system~\eqref{eq.vanhank}  can be regarded as a infinite
system of polynomial equations for the  `coordinates' $\f(a)$ of a point $\f$ in the infinite-dimensional vector space
$A^*$.

A conjectured statement is that the solutions of
equations~\eqref{eq.vanhank} give precisely the image of
$\Sym^{p|q}(X)$ in $A^*$. This would be an exact analog of the
Gelfand--Kolmogorov and Buchstaber--Rees theorems. The corresponding
algebraic statement should be a one-to-one correspondence between the
$p|q$-homomorphisms $A\to B$ and the algebra homomorphisms
$S^{p|q}A\to B$ extending the correspondence given by our formula~\eqref{eq.det}.

\bigskip\noindent
{\small
We thank V.~M.~Buchstaber for  discussions.
}

\end{document}